\documentclass[journal]{IEEEtran}

\usepackage{amsfonts, amsmath, amsthm, amssymb}
\usepackage{mathtools, graphicx, tikz, standalone, siunitx}
\usepackage{xcolor}

\newtheorem{lemma}{Lemma}
\newtheorem{prop}{Proposition}
\newtheorem{corollary}{Corollary}
\newtheorem{theorem}{Theorem}
\allowdisplaybreaks

\begin{document}

\title{Asynchronous Ad Hoc Networks with\\ Wireless Powered Cognitive Communications}

\author{\IEEEauthorblockN{Eleni Demarchou, \IEEEmembership{Student Member, IEEE}, Constantinos Psomas, \IEEEmembership{Senior Member, IEEE}, and Ioannis Krikidis, \IEEEmembership{Fellow, IEEE}}

\thanks{E. Demarchou, C. Psomas, and I. Krikidis are with the IRIDA Research Centre for Communication Technologies, Department of Electrical and Computer Engineering, University of Cyprus, Cyprus (e-mail: \{edemar01, psomas, krikidis\}@ucy.ac.cy).
  
This work was co-funded by the European Regional Development Fund and the Republic of Cyprus through the Research Promotion Foundation, under the projects INFRASTRUCTURES/1216/0017 and POST-DOC/0916/0256.}}

\maketitle

\begin{abstract}
Over the recent years, the proliferation of smart devices and their applications has led to a rapid evolution of the concept of the Internet of Things (IoT), advancing large scale machine type networks which are characterized by sporadic transmissions of short packets. In contrast to typical communication models and in order to capture a realistic IoT environment, we study an asynchronous channel access performed by a primary ad hoc network underlaid with a cognitive secondary wireless-powered ad hoc network. Specifically, we consider that the primary transmitters are connected to the power grid and employ asynchronous transmissions. On the other hand, the cognitive secondary transmitters have radio frequency energy harvesting capabilities, and their asynchronous channel access is established based on certain energy and interference based criteria. We model this sporadic channel traffic with time-space Poisson point processes and by using tools from stochastic geometry, we provide an analytical framework for the performance of this asynchronous system. In particular, we provide closed-form expressions for the information coverage probability and the spatial throughput for both networks and we derive the meta distribution of the signal-to-interference-plus-noise ratio. Finally, we present numerical results and provide important insights behind the main system parameters.
\end{abstract}

\begin{IEEEkeywords}
Asynchronous access, cognitive radio, wireless power transfer, Ad hoc networks, Poisson point process.
\end{IEEEkeywords}

\section{Introduction}
Over the recent years, there has been a rapid growth of the concept of the Internet of Things (IoT), connecting all kinds of different devices in order to improve the quality of life in various aspects. From system control and monitoring to data collection, IoT and machine to machine communications (M2M) cover a long range of operations, achieving to upgrade the systems management to more intelligent and more efficient \cite{IoT}. With all the benefits that have arisen, the applications and the development of smart devices are still growing. According to Cisco, by 2021 there will be 3.3 billion of M2M connections \cite{CISCO}, implying the development of networks with massive number of devices and sensors. However, one of the main challenges of such networks, is the power supply for those devices where conventional solutions such as batteries, require replacements or recharging which may be inconvenient or even infeasible. Therefore, the need for a practical and more efficient solution is mandatory for realizing these foreseen large scale networks.

Significant efforts have been devoted to the study of wireless power transfer (WPT), where devices are wirelessly powered by harvesting energy from electromagnetic radiation. Specifically, the devices are equipped with a rectifying antenna (rectenna) and a simple circuit for converting the radio frequency (RF) signals into direct current \cite{WPT}. Even though the energy conversion efficiency of WPT is limited, it is able to charge devices with low power requirements in a more flexible way, since charging can be controlled; in contrast to harvesting from renewable energy sources, where the available power for harvesting is unpredictable \cite{benefits WPT}. As such, WPT is a potential technology to meet the power source requirements of dense networks deemed by the applications of IoT. Towards this aim, the integration of WPT in communication systems has received significant attention by both industry and academia. The work in \cite{nonlinearities} sheds light on several input-output power models by considering the characteristics of RF energy harvesting systems and by taking into account the harvester's sensitivity and saturation levels. In order to maximize the harvested energy, the authors in \cite{TSP1} and \cite{tvt2} consider multiple antennas at the transmitter for energy beamforming. In the context of wireless powered communication networks (WPCN), the authors in \cite{RUI} introduce the protocol ``harvest-then-transmit'', where the users harvest RF energy from an access point during the downlink and transmit back information during the uplink. The same protocol is also applied in \cite{Karagiannidis}, where the authors maximize the network throughput by deriving the optimal power allocation of the access point's power as well as the time sharing among the wirelessly powered nodes which employ time-division-multiple-access (TDMA). By using tools from stochastic geometry, the authors in \cite{beacon} consider mobile devices which harvest energy from power beacons and transmit out-of-band RF signals; the uplink performance for a given outage constraint is investigated. The RF sources are also randomly distributed in the work \cite{Ginibre}, where the authors obtain the RF energy harvesting rate and provide appropriate upper bounds for the power and transmission outage probabilities. The authors in \cite{TVT} combine tools from both stochastic geometry and game theory in order to minimize the consumption of non-renewable energy and provide two energy trading approaches applicable in the energy harvesting small cell base stations (BSs).

Another challenge regarding the IoT and the employment of intelligent networks, is the spectrum allocation, where several types of devices with different applications establish communications by sharing a common medium. The emerging cognitive radio (CR) networks come out as a promising solution, since they allow the coexistence of a secondary network with a primary network, by dynamic spectrum access \cite{CR magazine}. This is achieved by spectrum sensing from the cognitive transmitters in order to decide whether spectrum access will not significantly affect the primary network \cite{spectrum sensing}. Specifically, in CR underlay networks, transmissions are established based on the permitted interference at the primary users \cite{CR}, \cite{Zhang}. The potential applications of CR, appear in various environments which deal with the spectrum congestion. The authors in \cite{Hossein1} consider macro BSs which are underlaid with CR femtocell BSs and they obtain the optimal spectrum sensing threshold which minimizes the outage probability of the femtocell BSs. By using stochastic geometry, the authors in \cite{LI} study a wireless powered underlay CR radio network, where secondary users harvest energy from RF transmissions in the downlink and transmit in the uplink by using the harvested energy with the TDMA scheme. Under an interference power constraint, they derive the optimal power control and the time allocation which maximizes the sum rate of the secondary users in the CR network. In this context, the authors in \cite{Hossein2} propose a random and a prioritized spectrum access policy for the coexistence of CR networks with primary cellular networks and they show that energy harvesting can be used with the CR transmission providing sufficient quality of service without degrading the primary network's performance. Furthermore, the work in \cite{RUI2} considers a similar energy harvesting approach for a CR secondary network, while assigning to each primary transmitter a guard zone to protect its receiver from the interference which occurs from the CR network. Under specific outage-probability constraints, the authors obtain the optimal transmission power and density of secondary transmitters which maximize the secondary network throughput. 

While the literature is rich in the study of WPCN applied in CR networks, in the classical communication frameworks the research community focus on coordinated and slotted transmissions by assuming perfect synchronization. This approach though cannot capture a realistic environment of a large scale M2M network \cite{Access2}. More specifically, in massive machine type networks, low power devices employ sporadic, event-driven transmissions of short packets \cite{machine1}. Consequently, to draw general insights on the performance of such large scale distributed networks, space randomness as well as uncoordinated transmissions should be taken into consideration. Towards in capturing the space randomness, a widely used stochastic pattern is the Poisson Point Process (PPP), which models the random locations of the network nodes \cite{HAE}. By extending the conventional PPP to time-space PPP (TS-PPP), also known as the Poisson Rain model, the randomness in both time and space can be taken into consideration \cite{Poisson Rain}. The authors in \cite{Aloha} adopt the TS-PPP and they use tools from stochastic geometry to provide the coverage probability for non-slotted Aloha wireless ad hoc networks. In addition, a similar approach is followed by the work \cite{full duplex1}, where an analytical framework of asynchronous large scale full-duplex networks is derived. In the context of WPCN, the random channel access is evaluated in \cite{asynchronous}, where the authors consider asynchronous backscatter communications with the radio frequency identification tag to be powered by RF harvesting. To the best of the authors knowledge, the asynchronous access in wireless powered cognitive communications has not yet been investigated. 

Motivated by the above, in this work, we investigate a cognitive secondary ad hoc network which is underlaid with a primary ad hoc network. In order to capture the characteristics of large scale M2M networks, we consider that the channel access in both networks is uncoordinated and unslotted. We apply the TS-PPP to model the random distribution of each ad hoc pair in both space and time. By using tools from stochastic geometry we provide a rigorous mathematical framework to evaluate the performance of each network. Specifically, the main contributions of this paper are summarized as follows:
\begin{itemize}
	\item We consider that the secondary transmitters have WPT capabilities and harvest energy from the asynchronous interfering transmissions in the network. We derive, for a predefined harvesting duration, the probability that the harvested energy exceeds a given energy threshold. In order to mitigate the overall network interference, each secondary transmitter employs a cognitive-based transmission scheme. In particular, we consider a protection area i.e., a guard zone around each secondary transmitter and if a primary receiver is located within this guard zone, then the secondary transmitter remains idle. By taking into account both the cognitive and energy capabilities, we derive in closed-form the probability of transmission for the secondary transmitters.
	\item Furthermore, for a predefined information transmission duration, we obtain closed-form expressions for the Laplace transform of the time-averaged interference for both networks. Using the Laplace transform and by taking into account all the active transmissions which are established over the information transmission duration, we derive the probability that the received signal strength attains a minimum threshold, i.e., the information coverage probability and we obtain the spatial throughput. 
	\item In addition, we apply the second order moment matching technique to represent the conditional coverage probability as a Beta random variable and we evaluate the distribution of the conditional information coverage probability. This distribution, also referred to as meta distribution in the literature, has been recently proposed as a fine-grained performance metric regarding the network links' reliability \cite{JSAC}.
	\item We present numerical results that validate our analysis and we discuss how the main system parameters affect the performance of the networks. We show that with long information and energy harvesting duration, more energy can be harvested at the secondary transmitters. However, even though long information transmissions improve the primary network's performance, they affect negatively the transmission probability and so the spatial throughput of the secondary network decreases. Finally, we show that with the proper combination of the density and the guard zone area size, the two networks can co-exist and achieve a decent quality of service. 
\end{itemize}

The rest of the paper is organized as follows. Section \ref{system} presents the system model and our main assumptions. Section \ref{secondary} provides the analytical framework for the derivation of the transmit probability for the secondary network. Section \ref{analysis} presents the analytical results for the coverage probabilities and provides the spatial throughput for each network. Finally, Section \ref{numer} presents the numerical results and Section \ref{Conclusion} concludes our work.

\underline{Notation:} $\mathbb{P}(X)$ represents the probability of the event $X$ with expected value $\mathbb{E}(X)$; $\mathbf{1}(x)$ is an indicator function which gives $1$ if $x$ is true, otherwise gives $0$; $\Im(x)$ returns the imaginary part of the complex number $x$ and $\imath = \sqrt{-1}$ denotes the imaginary unit; $\csc(\theta)$ is the cosecant of angle $\theta$ \cite{GRAD}.

\section{System Model}\label{system}
\begin{figure}[t]\centering
	\includegraphics[width=\linewidth]{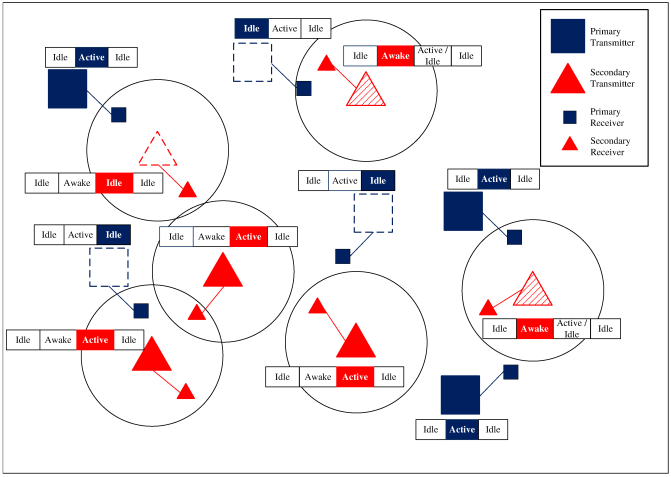}
	\caption{The primary ad hoc network underlaid with the secondary ad hoc network. The transmitters are captured along with their phases over time where the shaded ones indicate their current phase; circles represent the guard zone of the cognitive secondary transmitters.}\label{fig1}
\end{figure}
\subsection{Asynchronous network model}
We consider a primary ad hoc network underlaid with a cognitive secondary ad hoc network i.e., all nodes establish communications over the same channel. Both networks consist of a random number of transmitter-receiver pairs. Due to the low-power nature of the devices, pairing is achieved if the intended pair is located within a maximum distance \cite{Hossein2}. We assume a worst case scenario, where each transmitter is separated by its paired receiver by a distance $d$, which is common and fixed for all the pairs. The channel access is considered to be uncoordinated and asynchronous (e.g. un-slotted Aloha protocol \cite{Aloha}, \cite{asynchronous}) and each transmitter accesses the channel for a duration $T_I$. Note that, the transmitted packet length determines $T_I$ \cite{Aloha}. The primary transmitters are connected to the power grid and transmit with constant power $P_1$. Their operation consists of two phases: an idle phase (i.e. sleep phase) and an active phase (i.e. channel access phase). The randomness, both in time and space, is modeled for each network by an independent homogeneous TS-PPP \cite{Aloha}. We denote by $\Phi_1 = \{(x,t_x)\}$ the TS-PPP of density $\lambda_1$ modeling the primary transmitters, where $x \in \mathbb{R}^2$ is the location of the transmitter that initiates a transmission at time $t_x$. It is worth mentioning that the TS density, describes the sporadic channel traffic for a given area and time period, i.e., defines the number of nodes which are active at some random time and location \cite{full duplex1}. The secondary transmitters have RF energy harvesting capabilities and take cognition-based decisions regarding their operation. Therefore, while a primary transmitter's operation has two phases, a secondary transmitter undergoes three phases: an idle phase, an awake phase (i.e. RF harvesting/cognition phase), and an active phase. Specifically, during the awake phase, a secondary transmitter does the following: a) harvests RF ambient energy for a duration $T_E$, b) if the harvested energy exceeds a predefined threshold, the transmitter switches to active if an interference based criterion is satisfied (see Sections \ref{eh} and \ref{sensing} for details). The active secondary transmitters transmit with power $P_2$ which is a function of the harvested ambient RF energy. We denote by $\Phi_2 = \{(y,t_y)\}$ the TS-PPP of density $\lambda_2$ modeling the secondary transmitters, where $y \in \mathbb{R}^2$ is the location of the transmitter that enters the second phase at time $t_y$. Fig. \ref{fig1} schematically depicts our considered network model.

\subsection{Channel model}
All the wireless links in the network are considered to suffer from both small-scale block fading and large scale path-loss effect. The coherence time of the channel is assumed to be long enough such that all the signals experience flat Rayleigh fading. As such, the channel coefficient $h_x$ for the forward link from the primary transmitter located at $x \in \mathbb{R}^2$ to a point located at the origin is an exponentially distributed random variable with unit variance i.e., $h_x \sim \exp(1)$. Similarly, the channel coefficient $g_y$ from the secondary transmitter, located at $y \in \mathbb{R}^2$, to the origin is also exponentially distributed such that $g_{y} \sim \exp(1)$. We define as typical, the node which is located at the origin and leaves the idle phase at time zero. As such, we denote by $h_0$ and $g_0$ the channel coefficient of the typical primary and secondary link, respectively. The path-loss model is proportional to $(1+l^{\alpha})^{-1}$ where $l$ is the Euclidean distance between a transmitter and a receiver and $\alpha$ is the propagation exponent with $\alpha > 2$. Moreover, we denote by $r_x$ and $w_y$, the distances from the origin to a primary transmitter located at $x$ and a secondary transmitter located at $y$, respectively. Finally, all wireless links exhibit additive white Gaussian noise (AWGN) with variance $\sigma^2$.

\subsection{Energy harvesting model}\label{eh}
Consider a typical secondary transmitter located at the origin, which awakes at time zero and harvests ambient RF energy during the time interval $[0, T_E]$. The incident time-varying power available for harvesting can be expressed as\footnote{Note that, the amount of received power is dominated by the primary transmissions and the received power from secondary transmitters is considered negligible \cite{asynchronous}.} 
\begin{equation}\label{yt}
y(t) = P_1 \sum_{(x,t_x) \in \Phi_1} \phi(t,t_x) h_x (1+r_x^\alpha)^{-1},
\end{equation}
where $\phi(t,t_x) \triangleq \mathbf{1}(t_x \leq t \leq t_x+T_I)$ specifies whether the transmitter $\{(x,t_x)\} \in \Phi_1$ is active at time $t$. Therefore, the total harvested energy $E_H$ at the secondary transmitter during the harvesting period $T_E$ is
\begin{equation}
E_H = \int_0^{T_E} y(t)\,dt.
\end{equation}
In order to be able to power its operations, $E_H$ should be higher than a predefined threshold $\epsilon$ \cite{asynchronous}. By taking this constraint into account, along with the harvesting circuit's saturation level \cite{nonlinearities}, which we denote by $\mathcal{E}$, the input-output relationship model for the power $p$ at the typical secondary transmitter is expressed as follows
\begin{equation}\label{power}
p =
\begin{cases}
0, & E_H < \epsilon,\\
\frac{\mathbb{E}[E_H]}{T_I}, & \epsilon \leq E_H < \mathcal{E},\\
\frac{\mathcal{E}}{T_I}, & E_H \geq \mathcal{E},
\end{cases}
\end{equation}
where $\mathbb{E}[E_H]$ is the expected value of $E_H$. Note that, we consider the expected value of $E_H$ for the sake of analytical tractability \cite{asynchronous}.

We consider that the secondary transmitters employ the ``harvest-then-transmit'' scheme, where the harvested energy is lost unless is enough, i.e., $E_H > \epsilon$, to proceed with information transmission \cite{RUI}. In this case, a secondary transmitter can transmit, on average, with power $P_2$ which is given by
\begin{align} \label{transmit power}
P_2 = \mathbb{E}[p],
\end{align}
where $\mathbb{E}[p]$ is the expected value of $p$ over the three cases given in \eqref{power}. On the other hand, if sufficient energy is not harvested during the time interval $T_E$, then the secondary transmitter switches from awake back to idle mode.

\subsection{Network sensing model}\label{sensing}
Since all communications are established by accessing the same medium, interfering signals occur from each transmitter (primary or secondary) that occupies the channel. Therefore, the cognitive secondary transmitters control their transmissions to confine additive interference to nearby active primary receivers. As such, if sufficient energy has been harvested at a secondary transmitter i.e., $E_H > \epsilon$, it switches to the active phase and proceeds with information transmission if and only if the aggregate interference in the network is below a certain level. This criterion is defined as a guard zone around the transmitter. Therefore, if an active primary receiver falls within the guard zone, then the secondary transmitter decides not to transmit to its paired receiver. The guard zone is modeled by a disk of radius $\rho$ at the centre of each secondary transmitter as shown in Fig. \ref{fig1}.

\subsection{Information transmission model}
Recall that a primary transmitter switches randomly (in time) to the active phase for an interval $[0,T_I]$. On the other hand, a secondary transmitter switches to the active phase and proceeds with information transmission, if the criteria described in Sections \ref{eh} and \ref{sensing} hold. This asynchronous nature of channel access introduces a time-varying interference for the typical period $[0,T_I]$.

Consider the typical receiver (primary or secondary) located at the origin, whose paired transmitter is active at time zero. Interference is caused by any primary transmitter which is active as well as any secondary transmitter which has satisfied the criteria of the awake phase, within the interval $[0, T_I]$. Since the energy harvesting and cognition operations of the secondary transmitters depend on the primary transmitters, the processes $\Phi_1$ and $\Phi_2$ are not independent. However, for the sake of tractability, the secondary transmitters which are active during the information transmission period, are approximated by a thinned version of $\Phi_2$ \cite{HAE}, denoted by $\Phi_2^a$, with density $\lambda_2^a \triangleq \lambda_2 \pi_s$, where $\pi_s$ is the transmit probability of a secondary transmitter \cite{Hossein2}. Denote by $I_1(t)$ and $I_2(t)$ the time-varying interference which occurs from primary and secondary transmitters, respectively. Then, we have
\begin{align}
&I_1(t) = P_1 \sum_{(x,t_x) \in \Phi_1} \phi(t,t_x) h_x (1+r_x^{a})^{-1}\label{I1},\\
&I_2(t) = P_2 \sum_{(y,t_y) \in \Phi_2^a} \phi(t,t_y) g_y (1+w_y^{a})^{-1}\label{I2},
\end{align}
where $\phi(\cdot,\cdot)$ is given in Section \ref{eh}.

For the performance evaluation of the considered system, we will use the signal-to-interference-plus-noise ratio (SINR). Therefore, the received SINR at the typical primary receiver is given by
\begin{equation}\label{sir1}
\text{SINR}_1 = \frac{P_1 h_0 (1+d^a)^{-1}}{\mathrm{I}_1+\mathrm{I}_2+\sigma^2},
\end{equation}
where
\begin{align}\label{timeaveragedI1}
\mathrm{I}_1 = \frac{1}{T_I} \int_0^{T_I} I_1(t)\, dt,
\end{align}
and
\begin{align}\label{timeaveragedI2}
\mathrm{I}_2 = \frac{1}{T_I} \int_0^{T_I} I_2(t)\, dt,
\end{align}
is the aggregate interference caused by the primary and secondary transmitters, respectively, averaged over time \cite{full duplex1}.

Similar with the primary link, we take into account the asynchronous channel access and express the received SINR at the typical secondary receiver as a function of the time-averaged interference as follows
\begin{equation}\label{sir2}
\text{SINR}_2 = \frac{P_2 g_0 (1+d^a)^{-1}}{\mathrm{I}_1+\mathrm{I}_2+\sigma^2}.
\end{equation}

In addition, we consider that a receiver successfully decodes the received signal from its paired transmitter if the $\text{SINR}$ is above a predefined threshold $\zeta$. Furthermore, we assume that during the information transmission slot, the spectral efficiency is $R$ bits per channel use (bpcu) for all the links. As such, for the network performance evaluation our main metrics are the spatial throughput and the meta distribution of SINR.

\section{Wireless Powered Secondary Network}\label{secondary}
In this section, we derive the probability $\pi_s$ that a secondary transmitter is in active mode i.e., the energy and empty guard zone requirements are satisfied and thus the channel can be accessed by the secondary transmitter. For this purpose, we first evaluate the probability $\pi(\epsilon)$ that sufficient energy has been harvested by a secondary transmitter i.e., the complementary cumulative distribution function (CDF) of $E_H$. Next, we calculate the probability $\pi_\rho$ that no active primary receiver falls within the guard zone around a secondary transmitter.

\subsection{Energy Coverage}
We will now derive the energy coverage probability of a secondary transmitter. For this purpose, we will calculate the total harvested energy by taking into account which primary transmitters are active during the energy harvesting period i.e., the transmitters that initiated a transmission before the secondary transmitter switched to the awake mode, along with the transmitters that switched to active during the harvesting period.
Consider a secondary transmitter located at the origin during the interval $[0,T_E]$ i.e., during the awake phase. For the derivation of $\pi(\epsilon)$, we will use the characteristic function of the total harvested energy $E_H$, given in the following lemma.
\begin{lemma}\label{characteristic}
The characteristic function of the random variable $E_H$ is given by
\begin{align}
\mathcal{F}(z) = \exp\left(2\pi\lambda_1\int_{-T_I}^{T_E} \int_0^\infty \frac{\imath z P_1 \psi(t) u}{1 + u^\alpha - \imath z P_1 \psi(t)}\,du\,dt\right),
\end{align}
where $\psi(t)$ is given in Appendix \ref{app1} by \eqref{psi1}, \eqref{psi2} and \eqref{psi3}.
\end{lemma}

\begin{proof}
See Appendix \ref{app1}.
\end{proof}

We can now provide the complementary CDF of $E_H$, denoted by $\pi(\epsilon)$, in the following proposition.
\begin{prop}\label{E CDF}
The probability that a secondary transmitter has harvested energy higher than a predefined threshold $\epsilon$ during a time period $T_E$ is given by
\begin{align}\label{eps}
\pi(\epsilon) = \frac{1}{2} + \frac{1}{\pi} \int_0^\infty \frac{\Im[\exp(-\imath z\epsilon) \mathcal{F}(z)]}{z}\,dz,
\end{align}
where $\mathcal{F}(z)$ is given by Lemma \ref{characteristic}.
\end{prop}

\begin{proof}
The complementary CDF of $E_H$ is evaluated by
\begin{align}
\pi(\epsilon) &= \mathbb{P}(E_H > \epsilon) = 1 - \mathbb{P}(E_H < \epsilon)\nonumber\\
&= 1 - \left(\frac{1}{2} - \frac{1}{\pi} \int_0^\infty \frac{\Im[\exp(-\imath z \epsilon) \mathcal{F}(z)]}{z}\,dz\right),
\end{align}
which follows by applying the Gil-Pelaez Theorem \cite{Gil Pelaez}, where $\mathcal{F}(z)$ is the characteristic function of $E_H$ given by Lemma \ref{characteristic}, and the result follows.
\end{proof}
It is worth mentioning that, when the transmit power of the primary transmitters is very small i.e., $P_1\to0$ then $\pi(\epsilon)\to 0$. This follows from $\lim_{P_1\to 0} \mathcal{F}(z)=1$ and therefore
	 \begin{align}\label{pi/2}
	 \lim_{P_1\to\infty}& \int_0^\infty \frac{\Im[\exp(-\imath z \epsilon) \mathcal{F}(z)]}{z}\,dz= \int_{0}^{\infty} \frac{\Im[\exp(-\imath z \epsilon)]}{z}\,dz\nonumber\\
	 & = \int_0^\infty -\frac{ \sin(z \epsilon)}{z }\,dz=-\frac{\pi}{2}.
	 \end{align}
	The above is also true for $\lambda_1\to0$. On the other hand, as the transmit power of the primary transmitters increases i.e., $P_1\to\infty$, then $\pi(\epsilon)\to1$. This follows from 
	 \begin{equation}
	\lim_{P_1\to\infty} \int_0^\infty \frac{\Im[\exp(-\imath z \epsilon) \mathcal{F}(z)]}{z}\,dz=\frac{\pi}{2},
	\end{equation}
	which was validated numerically and is also valid for $\lambda_1\to\infty$.  As such, with higher transmit power or more transmitters available, a larger amount of energy can be absorbed by the secondary transmitters.
Next, we provide the probability $\pi_\rho$ that a primary receiver is not located inside a secondary transmitter's guard zone.

\subsection{Guard Zone}
In this section, we will obtain the probability that no active primary receiver is located within a secondary transmitter's guard zone. A secondary transmitter, by the end of the energy harvesting period, returns to the idle mode if there is not sufficient energy, otherwise it decides whether to initiate information transmission to its paired receiver based on the locations of the active primary receivers. If an active primary receiver is located within a distance higher than $\rho$ from the secondary transmitter then the transmission begins. Note that, the decision is considered to be taken instantaneously right after the energy harvesting period. That is, the decision takes into account the active primary pairs which have already initiated their transmissions by the end of the harvesting period.

Specifically, a secondary transmitter $(y,t_y)$ senses the active pairs which initiated their transmissions during the time interval $[t_y-T_I,t_y]$. As we consider a finite time period of network observation i.e., a unit time window, the initiation of a transmission is uniformly distributed over time. Therefore, the probability that a primary pair is active at time $t_y$ is simply $t_y-(t_y-T_I)=T_I$. As such, the primary transmitters which are active at time $t_y$ form a thinned version of TS-PPP $\Phi_1$ with density $\lambda_1T_I$. Now by considering the protection area around the secondary transmitter, we provide the probability $\pi_\rho$ in the following proposition.

\begin{prop}\label{rho}
The probability that no active primary receiver is located within a secondary transmitter's guard zone is given by
\begin{equation}
\pi_\rho = \exp(-\lambda_1 T_I \pi \rho^2), \label{rho Eq}
\end{equation}
where $\rho$ is the guard zone's radius.
\end{prop}

\begin{proof}
Recall that each receiver is separated from its paired transmitter by a common distance $d$. By leveraging the homogeneity of the PPP and applying the displacement theorem \cite{HAE}, we can consider that the primary receivers are also distributed according to a PPP of density $\lambda_1$. Thus, the density of the active primary receivers is $\lambda_1T_I$. Now, let $\mathcal{N}(A)$ denote the number of active primary receivers in the area $A$. Then, the probability that $n$ receivers are located within a guard zone area of area $\pi\rho^2$ is a Poisson random variable with parameter $\lambda_1 T_I \pi \rho^2$ and is evaluated by
\begin{equation}
\mathbb{P}\{\mathcal{N}(\pi\rho^2) = n\} = \frac{(\lambda_1 T_I \pi \rho^2)^n \exp(-\lambda_1 T_I \pi \rho^2)}{n!}.
\end{equation}
The result follows by considering the void probability i.e., $n = 0$.
\end{proof}
From (\ref{rho Eq}), it follows that if $\lambda_1\to \infty$, then $\pi_{\rho}\to 0$, therefore the employment of a very dense primary network does not permit the coexistence of a cognitive secondary network since the likelihood of an empty guard zone becomes zero. Similarly, if $\rho\to \infty$, then $\pi_{\rho}\to 0$ which implies that a large guard zone protects the primary network by completely blocking the secondary network. On the other hand, if $\lambda_1\to 0$ then $\pi_{\rho}\to 1$ i.e., since no primary pairs are employed, the guard zone is always empty. Finally, with $\rho\to 0$ the cognition capability of the secondary network is completely ignored which results in $\pi_{\rho}\to 1$.

\subsection{Active Secondary Transmitters}
In order to study the interaction between primary and secondary pairs, we need to investigate the asynchronous channel access by both primary and secondary transmitters. According to the network model, a secondary transmitter has channel access if enough energy has been harvested during $T_E$ and if no active primary receiver is located within a secondary transmitter's guard zone. Therefore, the transmit probability i.e., the probability that a secondary transmitter initiates a transmission is given by
\begin{equation}
\pi_s = \pi(\epsilon) \pi_\rho,
\end{equation}
where $\pi(\epsilon)$ and $\pi_\rho$ are provided in Propositions $\ref{E CDF}$ and $\ref{rho}$, respectively. From the above discussion regarding the behaviour of $\pi(\epsilon)$ and $\pi_{\rho}$, it follows that the likelihood of a secondary transmitter to access the channel experiences a trade-off regarding the density of the primary transmitters. Specifically, the employment of a very dense primary network implies that the energy threshold is satisfied at the cost of a non-empty guard zone. On the other hand, with a low primary transmitters density, exceeding the energy threshold is harder while the empty guard zone requirement has higher likelihood to be achieved. The same behaviour is valid for the effect of $T_I$ on the transmit probability. Furthermore, we provide the following proposition regarding the average transmit power of the secondary transmitters.

\begin{prop}\label{sec_power}
An active secondary transmitter transmits with average power
\begin{align}
P_2 = \frac{2 \pi^2 \lambda_1 T_E P_1}{\alpha} \csc\left(\frac{2\pi}{\alpha}\right) (\pi(\epsilon)-\pi(\mathcal{E})) + \frac{\mathcal{E}}{T_I} \pi(\mathcal{E}),
\end{align}
where $\pi(\mathcal{E})$ denotes the probability that the total harvested energy $E_H$ exceeds the saturation level $\mathcal{E}$ and is given in Proposition \ref{E CDF}.
\end{prop}

\begin{proof}
See Appendix \ref{E EH}.
\end{proof}
It can be observed that, in order to increase the transmit power $P_2$, one has to increase the harvesting or information transmission duration so that to absorb more power from the primary transmitters. This is because by increasing $T_I$ or $T_E$, a primary transmitter transmits or a secondary receiver harvests energy for a longer time period, respectively. However, for achieving higher $P_2$ an increase in $T_E$ instead of $T_I$ would be preferable since the transmit power $P_2$ is proportional to $1/T_I$. This is clear from the asymptotic scenario $P_1\to\infty$. As discussed before, when $P_1\to\infty$, then $\pi(\epsilon)=\pi(\mathcal{E})=1$ and therefore the transmit power of the secondary transmitters becomes constant at
	\begin{equation}
		\lim_{P_1\to\infty}P_2=\frac{\mathcal{E}}{T_I}.
	\end{equation}
	This is expected since the rectification circuit at the secondary transmitters has a saturation level at $\mathcal{E}$ which limits the harvested energy and consequently the transmit power of the secondary transmitters.
  
\section{Performance Evaluation}\label{analysis}
We now turn our attention on deriving the main analytical expressions for the performance evaluation of our considered network, namely, the spatial throughput and the meta distribution.

\subsection{Coverage Probability \& Spatial Throughput}
The spatial throughput is defined as the information successfully transmitted per unit time and unit area. As such, we first need to derive the coverage probability i.e., the probability that the received SINR is higher than a predefined threshold $\zeta$. We denote by $p_1^c$ and $p_2^c$, the information coverage probability for the typical primary and secondary receiver, respectively. We consider a typical primary receiver and derive the coverage probability $p_1^c = \mathbb{P}(\text{SINR}_1 \geq \zeta)$ in the following proposition.

\begin{theorem}\label{thm1}
The information coverage probability of the typical primary receiver for a threshold $\zeta$, is given by
\begin{align}\label{coverP1}
p_1^c(\zeta) = \mathcal{L}_{I_1}(s) \mathcal{L}_{I_2}(s) \exp(-\sigma^2 s),
\end{align}
where $s = \frac{\zeta (1+d^{\alpha})}{P_1}$, $\mathcal{L}_{I_1}(s)$ and $\mathcal{L}_{I_2}(s)$ denote the Laplace transform of the interference from the primary and secondary transmitters, respectively, and are given by
\begin{align}\label{laplace}
\mathcal{L}_{I_{n}}(s) = \exp\Bigg(&-2 \lambda_n \pi^2 T_I \csc\left(\frac{2\pi}{\alpha}\right)\nonumber\\
&\times \frac{(1+P_n s)^{2/\alpha}(2 P_n s-\alpha)+\alpha}{P_n s (2+\alpha)}\Bigg),
\end{align}
where $n \in \{1,2\}$, $P_1$ is the transmit power of the primary transmitters and $P_2$ is the average transmit power of the secondary transmitters given in Proposition \ref{sec_power}.
\end{theorem}

\begin{proof}
See Appendix \ref{app2}.
\end{proof}
It is worth mentioning that, in the case of $P_1\to\infty$, the Laplace transforms of the interference from primary and secondary transmitters reduce to
	\begin{align}
	&\lim_{P_1\to\infty}	\mathcal{L}_{I_1}=\exp\Bigg(-\frac{2 \pi ^2 \lambda_1 T_I \csc \left(\frac{2 \pi }{\alpha}\right)}{(\alpha+2) \zeta \left(d^\alpha+1\right)}\nonumber\\
	&\times  \left( \left(\zeta d^\alpha+\zeta+1\right)^{2/\alpha}\left(2 \zeta \left(d^\alpha+1\right)-\alpha\right) +\alpha\right)\Bigg),
	\end{align}
	and
	\begin{equation}
		\lim_{P_1\to\infty}	\mathcal{L}_{I_2}=1,
	\end{equation}
	respectively. As such, the coverage probability of the primary transmitter is independent of the secondary network. Furthermore, for a given threshold $\zeta$, the coverage probability is inversely proportional to the density $\lambda_1$ and the path-loss parameters $\alpha$ and $d$.
We now consider the typical secondary receiver in order to derive the coverage probability $p_2^c$. Note that, we consider that the typical pair at time zero is in active mode, that is, enough energy has been harvested and there is no active primary receiver in the secondary transmitter's guard zone.
\begin{prop}\label{coverP2}
The information coverage probability of the typical primary receiver for a threshold $\zeta$, is given by
\begin{align}
p_2^c(\zeta) = \mathcal{L}_{I_1}(s) \mathcal{L}_{I_2}(s) \exp(-\sigma^2 s),
\end{align}
where $s = \frac{\zeta (1+d^{\alpha})}{P_2}$, $\mathcal{L}_{I_1}(s)$ and $\mathcal{L}_{I_2}(s)$ are given by \eqref{laplace}.
\end{prop}

\begin{proof}
The proof is similar to the derivation of $p_1^c$ and thus is omitted.
\end{proof}
We now focus on the coverage probability of a secondary receiver for the case where $P_1\to\infty$. In this case, the Laplace transform of the interference from the primary transmitters converges to zero i.e., $\lim_{P_1\to\infty} \mathcal{L}_{I_1}(s)=0$ and therefore imposes zero coverage for the secondary receivers.
	
Provided with the coverage probabilities, we can now examine the performance of the asynchronous networks by evaluating each network's spatial throughput which expresses the amount of information which has been successfully transmitted over a time period $T_I$ with a given rate $R$ over a given area and transmitters' density. Therefore, similar to \cite{asynchronous}, the spatial throughput of the primary and secondary networks for a given threshold $\zeta$ are given by
\begin{align}
\mathcal{T}_1 = T_I R \lambda_1 p_1^c(\zeta) ~ \text{(bpcu)},
\end{align}
and
\begin{align}
&\mathcal{T}_2 = T_I R \lambda_2^a p_2^c(\zeta) ~ \text{(bpcu)},
\end{align}
where $p_1^c(\zeta)$ and $p_2^c(\zeta)$ are given by Theorem \ref{thm1} and Proposition \ref{coverP2}, respectively.

\subsection{Meta Distribution}
In what follows, we derive the SINR's meta distribution to obtain the fraction of users that achieve a certain information coverage probability. Mathematically, this is defined as the distribution of the conditional coverage probability for an SINR threshold $\zeta$ over a given TS-PPP. For the derivation of the meta distribution we will use the reduced Palm probability $\mathbb{P}^!_o(\cdot)$, which is the conditional probability of an event given that the typical receiver exists at a specific location i.e., the typical link is active \cite{HAE}. Therefore, the meta distribution of the SINR can be expressed as \cite{HAE2}
\begin{equation}
F_n(x) = \mathbb{P}^!_o(q_n^c(\zeta) > x),
\end{equation}
where $x \in [0,1]$ and $q_n^c(\zeta)$ is the conditional coverage probability for a given threshold $\zeta$ of the typical primary/secondary receiver, where $n \in \{1,2\}$ i.e., $q_n^c(\zeta) = \mathbb{P}(\text{SINR}_n > \zeta ~\big|~ \Phi_n)$. Since $\Phi_n$ is given, $q_n^c(\zeta)$ can be written as
\begin{align}\label{qfun}
&q_n^c(\zeta) = \left[\prod_{(x,t_x) \in \Phi_1} \frac{1}{1 + s P_1 \chi(t_x) (1+r_x^{\alpha})^{-1}} \right]\nonumber\\ &\times\left[\prod_{(y,t_y) \in \Phi_2^a} \frac{1}{1 + s P_2 \chi(t_y) (1+w_y^{\alpha})^{-1}} \right] \exp(-\sigma^2 s),
\end{align}
which follows similarly to the proof of $p_n^c(\zeta)$ given in Appendix \ref{app2} and where $s = \frac{\zeta (1+d^{\alpha})}{P_1}$ for $n = 1$ and $s = \frac{\zeta (1+d^{\alpha})}{P_2}$ for $n = 2$.

As in \cite{HAE2}, the derivation of the SINR's meta distribution $F_n(x)$, requires a second order moment matching technique as to express the conditional coverage probability $q_n^c$ as a Beta random variable. This technique is expressed in following lemma.
\begin{lemma}\label{beta}
A random variable $X$ defined in the interval $[0,1]$, with first and second moments $\mathbb{E}\left[X\right]$ and $\mathbb{E}\left[X^2\right]$, can be represented as a Beta random variable $B(\gamma, \delta)$ with shape parameters given by
\begin{align}
\gamma = \frac{\mathbb{E}[X] \mathbb{E}[X^2] - \mathbb{E}[X]^2}{\mathbb{E}[X]^2 - \mathbb{E}[X^2]},
  \end{align}
\begin{align}
\delta = \frac{(1 - \mathbb{E}[X]) (\mathbb{E}[X^2] - \mathbb{E}[X])}{\mathbb{E}[X]^2-\mathbb{E}[X^2]}.
  \end{align}
\end{lemma}
Therefore, we require the first and second moment of $q_n^c(\zeta)$ i.e., $\mathbb{E}[q_n^c(\zeta)]$ and $\mathbb{E}[q_n^c(\zeta)^2]$. Note that, the first moment of $q_n^c(\zeta)$ is given directly by the coverage probability of $p_n^c(\zeta)$ in Theorem \ref{thm1} and Proposition \ref{coverP2}. As such, we only need to evaluate the second moment of $q_n^c$, provided below.
\begin{corollary}
The second moment of $q_n^c$, $\mathbb{E}[q_n^c(\zeta)^2]$, is
\begin{align}
\mathbb{E}[q_n^c(\zeta)^2] = \mathcal{L}^2_{I_2}(s)\mathcal{L}^2_{I_2}(s)\exp(-2\sigma^2 s),
\end{align}
where
\begin{align}
&\mathcal{L}^2_{I_n}(s) = \exp\Bigg(-4\pi\lambda_n\nonumber\\
&\times\! \int_0^{T_I}\!\!\! \int_0^\infty\! \left(1-\left(\frac{1}{1 + s P_n \chi(t) (1+u^{\alpha})^{-1}}\right)^2\right)u\, du\,dt\Bigg),
\end{align}
where $s$ is defined as above.
\end{corollary}

\begin{proof}
The result follows directly from \eqref{qfun} and the proof of Theorem \ref{thm1}.
\end{proof}

Finally, having the first two moments of $q_n^c$, we can approximate it as a Beta random variable with shape parameters as in Lemma \ref{beta}. Then, the SINR's meta distribution is evaluated by applying the complementary CDF of a Beta random variable given by \cite{HAE2}
\begin{equation}
F_n(x) = 1 - \frac{1}{\beta(\gamma,\delta)}\int_0^x z^{\gamma-1} (1 - z)^{\delta-1}\,dz,
\end{equation}
where $\beta(\gamma,\delta)$ denotes the Beta function.

\section{Numerical Results}\label{numer}
In this section, we present our numerical results to evaluate the performance of both the primary and secondary networks. Unless otherwise stated, in our results we consider the following parameters: $d=1$ \si{\meter}, $\sigma^2=-50$ \si{\dB m}, $\alpha=3$, $\lambda_1=\SI[per-mode=symbol]{0.1}[ ]{\per\square\metre\per\second}$, $\lambda_2=\SI[per-mode=symbol]{1}[ ]{\per\square\metre\per\second}$, $P_1=1$ \si{\watt}, $\epsilon=0.1$ \si{\joule}, $\mathcal{E}=0.5$ \si{\joule} and $R=\log(1+\zeta)$ bpcu. 

\begin{figure}[t]\centering
\includegraphics[width=0.8\linewidth]{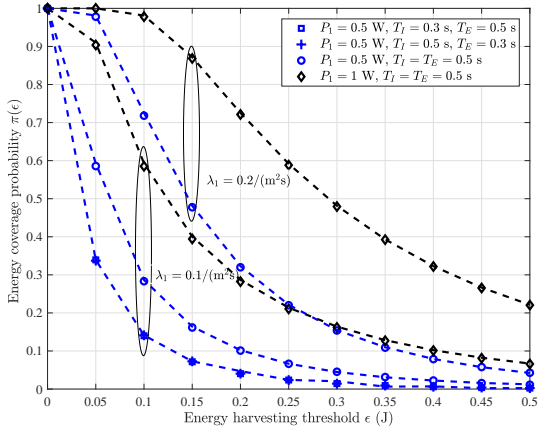}
\caption{Energy coverage probability versus energy harvesting threshold; markers and dashed lines correspond to simulation and analytical results, respectively.}\label{Ecoverage}
\end{figure}

Fig. \ref{Ecoverage} plots the energy coverage probability $\pi(\epsilon)$ versus the energy harvesting threshold $\epsilon$, where $\pi(\epsilon)$ is plotted for thresholds lower than the energy saturation level $\mathcal{E}$, after which the energy coverage is zero. We present the energy coverage probability for $P_1 = \{0.5, 1\}$ W, $T_I = \{0.3, 0.5\}$, $T_E = \{0.3, 0.5\}$ and $\lambda_1 = \{0.1, 0.2\}/(\si{\square\meter}\si{\second})$. In all cases, as the harvesting threshold increases, the energy coverage probability decreases. As can be seen, when $P_1 = 0.5$ W and $\lambda_1 = 0.1/(\si{\square\meter}\si{\second})$, the energy coverage probability is the same for the two cases of $T_I = 0.3\, \si{\second}$, $T_E = 0.5\, \si{\second}$ and $T_I = 0.5\, \si{\second}$, $T_E = 0.3\, \si{\second}$. This is due to the fact that, the amount of harvested energy is a function of the product $T_I T_E$, (see \eqref{mean}), and therefore by keeping $T_I T_E$ constant, would not affect the energy coverage probability. This is also shown in the case with $T_I = T_E = 0.5$ \si{\second}, which results in a higher $\pi(\epsilon)$. In addition, by increasing the power of the primary transmitters to $P = 1$ W and their density $\lambda_1$, the energy coverage probability is improved since in both cases, the power available for energy harvesting is higher. It can be observed, that at lower thresholds, an increase in the density i.e., $P_1 = 0.5$ W, $\lambda_1 = 0.2/(\si{\square\meter}\si{\second})$ performs better than increasing the power i.e., $P_1 = 1$ W, $\lambda_1 = 0.1/(\si{\square\meter}\si{\second})$, while for $\epsilon > 0.3$ J better performance is obtained with more power rather than higher density.

\begin{figure}[t]\centering
\includegraphics[width=0.8\linewidth]{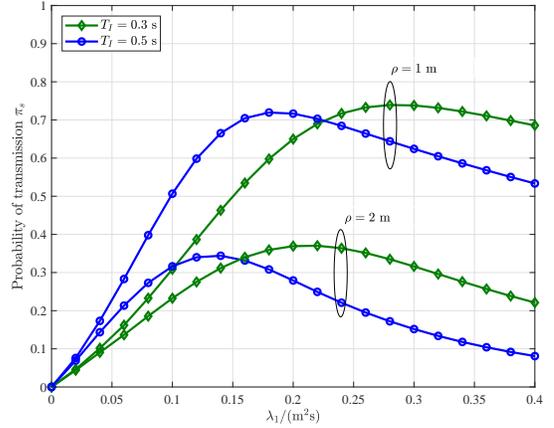}
\caption{Secondary transmitters' transmit probability versus the primary transmitters' density; $P_1 = 1$ \si{\watt}, $T_E = 0.5$ s and $\epsilon = 0.1$ \si{\joule}.}\label{PS}
\end{figure}

\begin{figure}[t]\centering
	\includegraphics[width=0.8\linewidth]{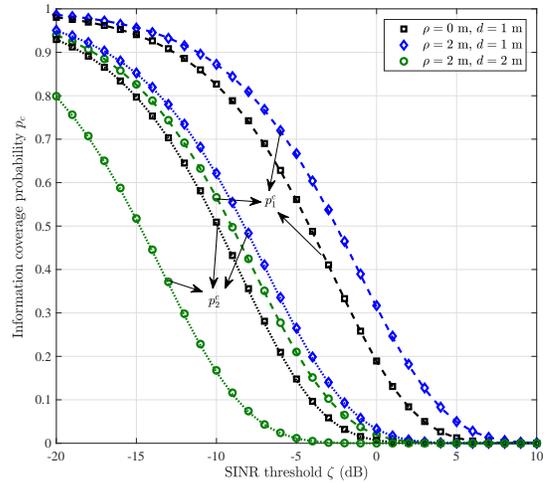}
	\caption{Information coverage probability versus the SINR threshold; markers and dashed lines correspond to simulation and analytical results respectively. $\lambda_1=\SI[per-mode=symbol]{0.1}[ ]{\per\square\metre\per\second}$, $\lambda_2=\SI[per-mode=symbol]{1}[ ]{\per\square\metre\per\second}$, $P_1=1$\si{\watt}, $\epsilon=0.1$ \si{\joule} and $\mathcal{E}=0.5$ \si{\joule}.}\label{cover}
\end{figure}
In Fig. \ref{PS}, the transmission probability $\pi_s$ for the secondary transmitters is depicted with respect to the primary transmitters' density for $\rho = \{1, 2\}$ \si{\meter} and $T_I = \{0.3, 0,5\}$ s. By comparing the cases with $\rho = 1$ \si{\meter} and $\rho = 2$ \si{\meter}, we can see that with a higher radius $\rho$ of the guard zone, the probability of transmission is lower. This is due to the fact that, a higher $\rho$ implies a lower probability for an empty guard zone $\pi_\rho$ and subsequently $\pi_s$ is lower. In addition, in all the considered plots, we can see that as the density $\lambda_1$ increases, $\pi_s$ increases until it reaches to a peak value after which it starts decreasing again. Recall from Fig. \ref{Ecoverage}, that a lower density implies a lower $\pi(\epsilon)$, whereas for low densities, $\pi_\rho$ is high due to sparser spatial distribution of the transmitters. As such, as the density increases, $\pi(\epsilon)$ increases rapidly from zero towards $\pi(\epsilon) = 1$, which results in a better $\pi_s$. On the other hand, as the density increases while the energy coverage probability remains $\pi(\epsilon) = 1$, the probability $\pi_\rho$ decreases, which justifies the decrease in $\pi_s$ after its peak value. Furthermore, when comparing the cases with $T_I = \{0.3, 0.5\}$ \si{\second}, we can see that a higher $T_I$ leads to a faster growth of $\pi_s$, which occurs from the highest $\pi(\epsilon)$ and therefore $\pi_s$ reaches its peak value at a lower density. However, after the peak value of $\pi_s$, the decrease is more sharp in the case of $T_I = 0.5$ \si{\second}. This is expected since a larger $T_I$ imposes a longer active phase and thus the secondary transmitter's decision to transmit depends on a denser network implying a lower $\pi_\rho$, (see \eqref{rho Eq}).

\begin{figure}[t]\centering
	\includegraphics[width=0.8\linewidth]{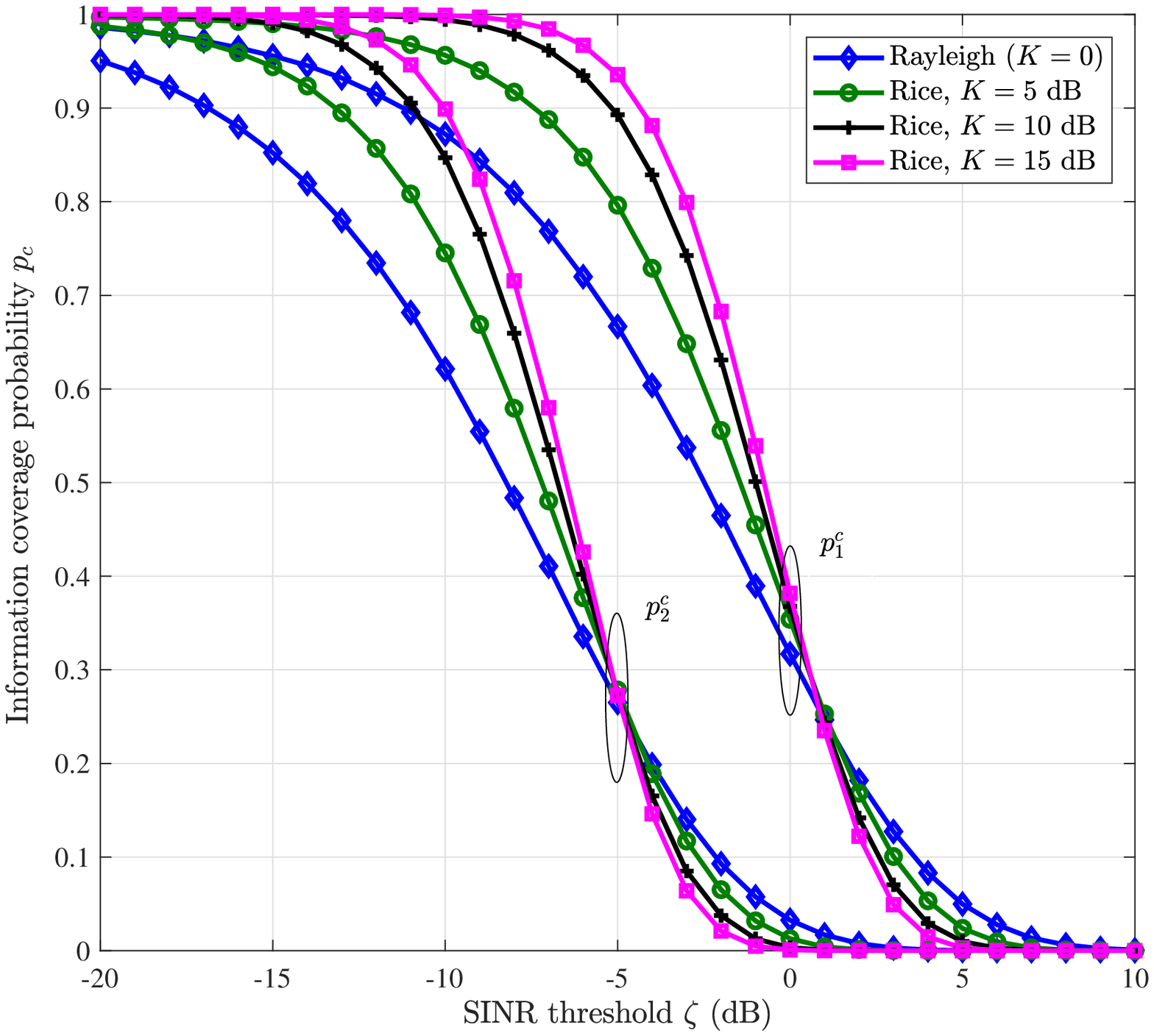}
	\caption{Information coverage probability versus the SINR threshold with Rician fading between the pairs. $\lambda_1=\SI[per-mode=symbol]{0.1}[ ]{\per\square\metre\per\second}$, $\lambda_2=\SI[per-mode=symbol]{1}[ ]{\per\square\metre\per\second}$, $d=1$ m, $\rho=2$ m, $P_1=1$\si{\watt}, $\epsilon=0.1$ \si{\joule} and $\mathcal{E}=0.5$ \si{\joule}.}\label{rice}
\end{figure}
\begin{figure}[t]\centering
	\includegraphics[width=0.8\linewidth]{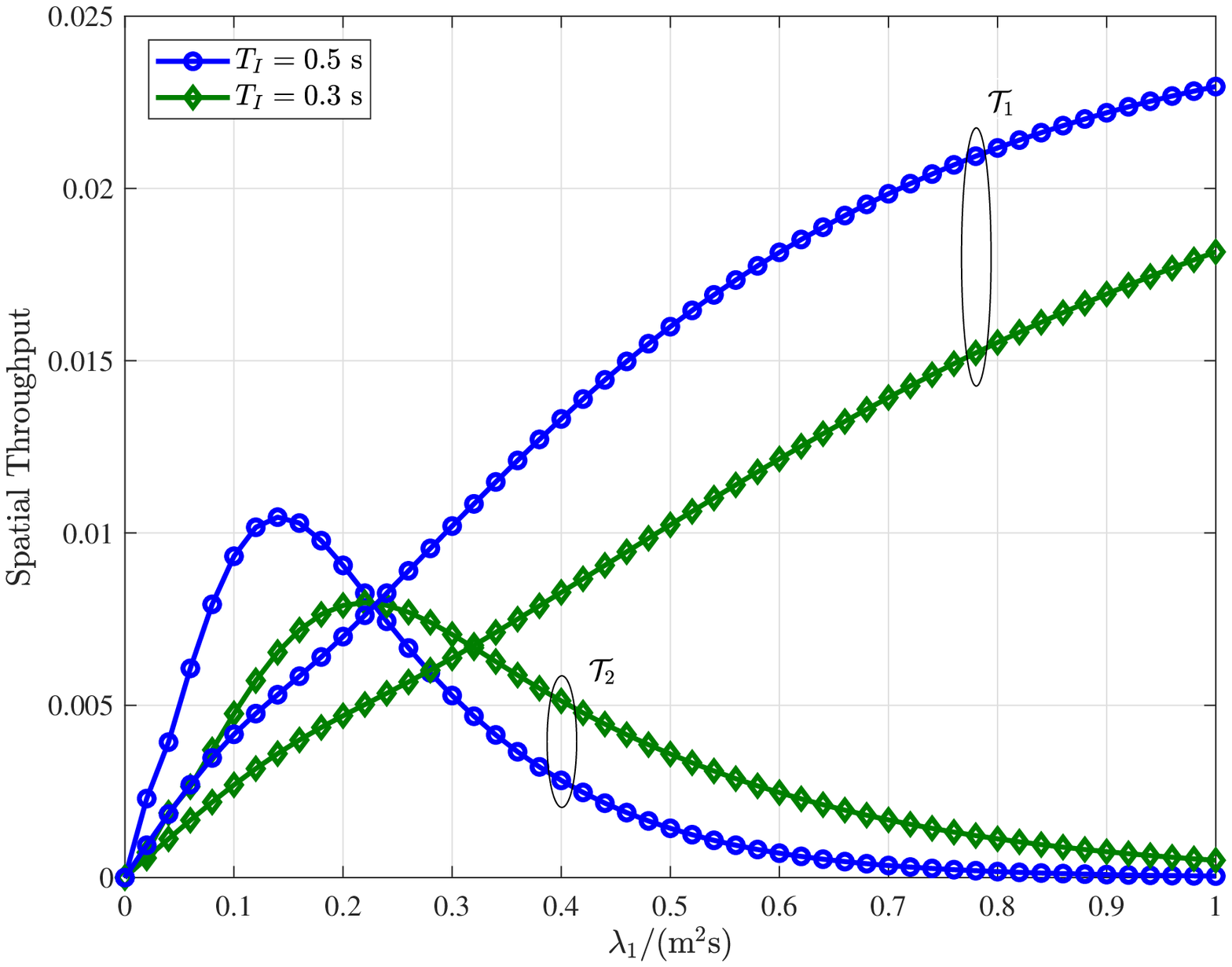}
	\caption{Spatial Throughput versus the primary transmitters' density; $\zeta=-10$ dB, $\lambda_2=\SI[per-mode=symbol]{1}[ ]{\per\square\metre\per\second}$, $P_1=1$ \si{\watt}, $T_E=0.5$ s, $\epsilon=0.1$, \si{\joule}, $\mathcal{E}=0.5$ \si{\joule} and $\rho=2$ m.}\label{Thr}
\end{figure}
Fig. \ref{cover} illustrates the coverage probability versus the SINR threshold $\zeta$ for the primary and secondary receivers with and without cognition i.e., $\rho = 2$ \si{\meter} and $\rho = 0$ \si{\meter} for $d=1$ m. As expected the coverage probability decreases as the SINR threshold increases. In addition, for both the primary and the secondary receivers, when no cognition is considered ($\rho = 0$ \si{\meter}), implies that $\pi_\rho = 1$, and therefore the density of the secondary transmitters is higher. As a result, the interference from the secondary network is stronger, which results in a lower coverage probability. On the other hand, when the guard zone radius is $\rho > 0$ \si{\meter}, a lower $\pi_\rho$ is obtained. This subsequently implies a lower $\pi_s$ and thus the interference level is lower. Furthermore, while the gain in coverage probability that the two receivers have with cognition is the same, we can see a lower coverage probability at the secondary receiver. This is due to the fact that the interference for both the receivers is the same, while the transmit power of the primary transmitter is higher than the one of the secondary transmitter i.e., $P_1 > P_2$ which implies a better performance for the primary receiver. In addition, as can be seen, when we increase the distance between the pairs i.e., $d=2$, the coverage probability, for both the primary and the secondary receiver, decreases since the path-loss attenuation is higher. The coverage probability is also demonstrated in Fig. \ref{rice} where the channel fading between the pairs is considered to follow Rician distribution. We present simulation results for different Rician factors $K$ and compare with the case where Rayleigh fading is assumed ($K=0$). As can be seen, at low SINR thresholds the highest the Rician factor $K$, the better performance occurs for both the primary and the secondary user. On the other hand, at high SINR thresholds as $K$ increases, the coverage probability becomes lower whereas the highest coverage probability is obtained over Rayleigh fading.

\begin{figure}[t]\centering
	\includegraphics[width=0.8\linewidth]{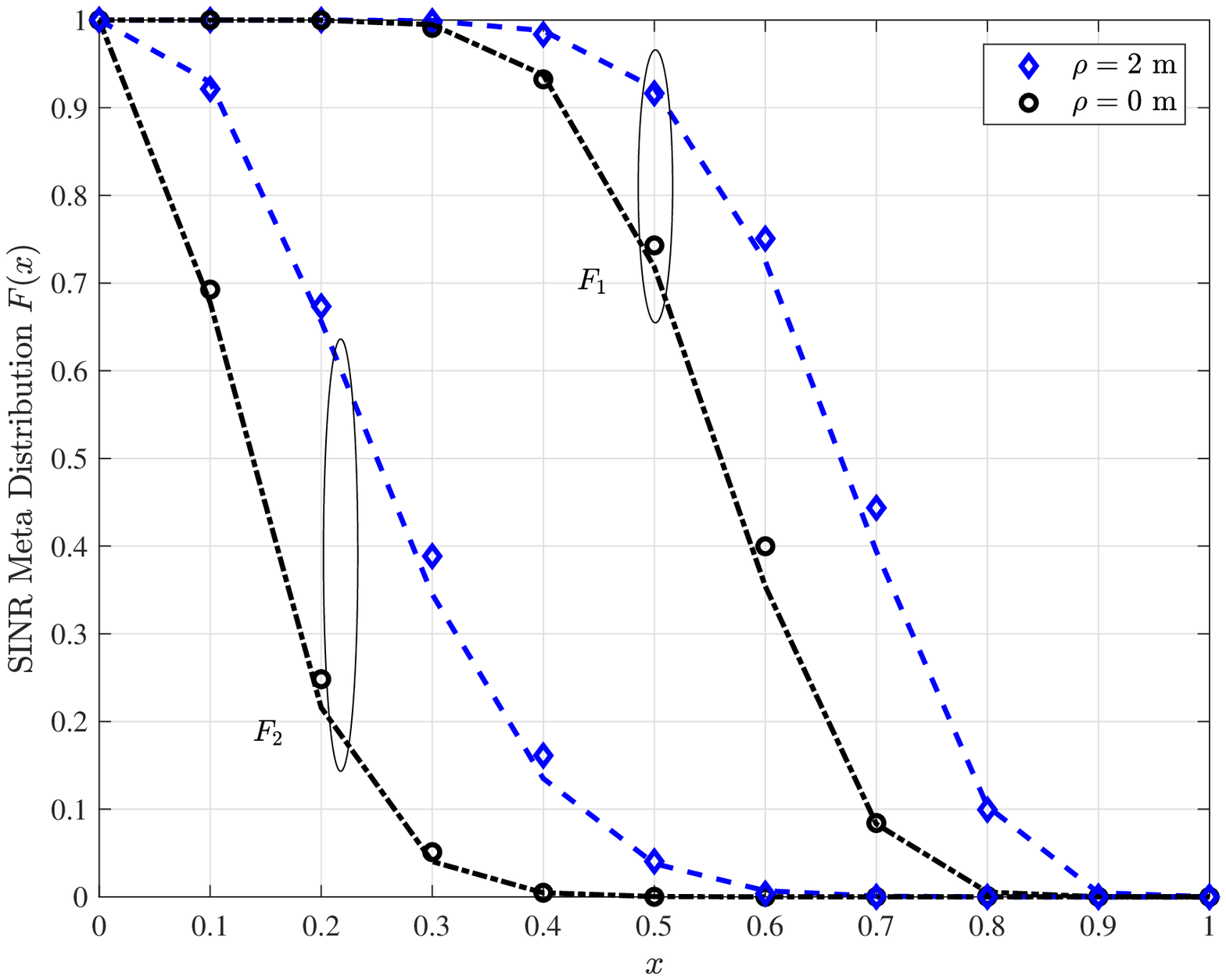}
	\caption{SINR meta Distribution versus threshold x; markers and dashed lines correspond to simulation and analytical results respectively. $\zeta=-5$ dB, $\lambda_1=\SI[per-mode=symbol]{0.1}[ ]{\per\square\metre\per\second}$, $\lambda_2=\SI[per-mode=symbol]{1}[ ]{\per\square\metre\per\second}$, $d=1$ m, $P_1=1$\si{\watt}, $\epsilon=0.1$ \si{\joule} and $\mathcal{E}=0.5$ \si{\joule}.}\label{meta}
\end{figure}

In Fig. \ref{Thr}, we present the spatial throughput of each of the two networks, primary and secondary, versus the primary transmitters' density for $T_I = \{0.3, 0.5\}$ \si{\second}. Regarding the spatial throughput of the secondary transmitters $\mathcal{T}_2$, we can make similar observations as with the probability of transmission $\pi_s$ (see Fig. \ref{PS}). This is because $\mathcal{T}_2$ is a linear function of $\pi_s$ and therefore $\mathcal{T}_2$ follows the same behavior for both $T_I = \{0.3, 0.5\}$ \si{\second}. On the other hand, the spatial throughput of the primary network increases as the density of the primary transmitters increases, which is also expected since $\mathcal{T}_1$ is a linear function of $\lambda_1$. This is also valid for the values of $T_I$; by using a higher $T_I$, a better throughput is achieved. In addition, we can see that the increase in $\mathcal{T}_1$ starts rapidly and after some point the increase occurs more slowly. This is due to the fact that the spatial throughput is also a function of the coverage probability, where a higher $\lambda_1$ implies higher interference power and hence a lower coverage probability.

\begin{figure}[t]\centering
	\includegraphics[width=0.8\linewidth]{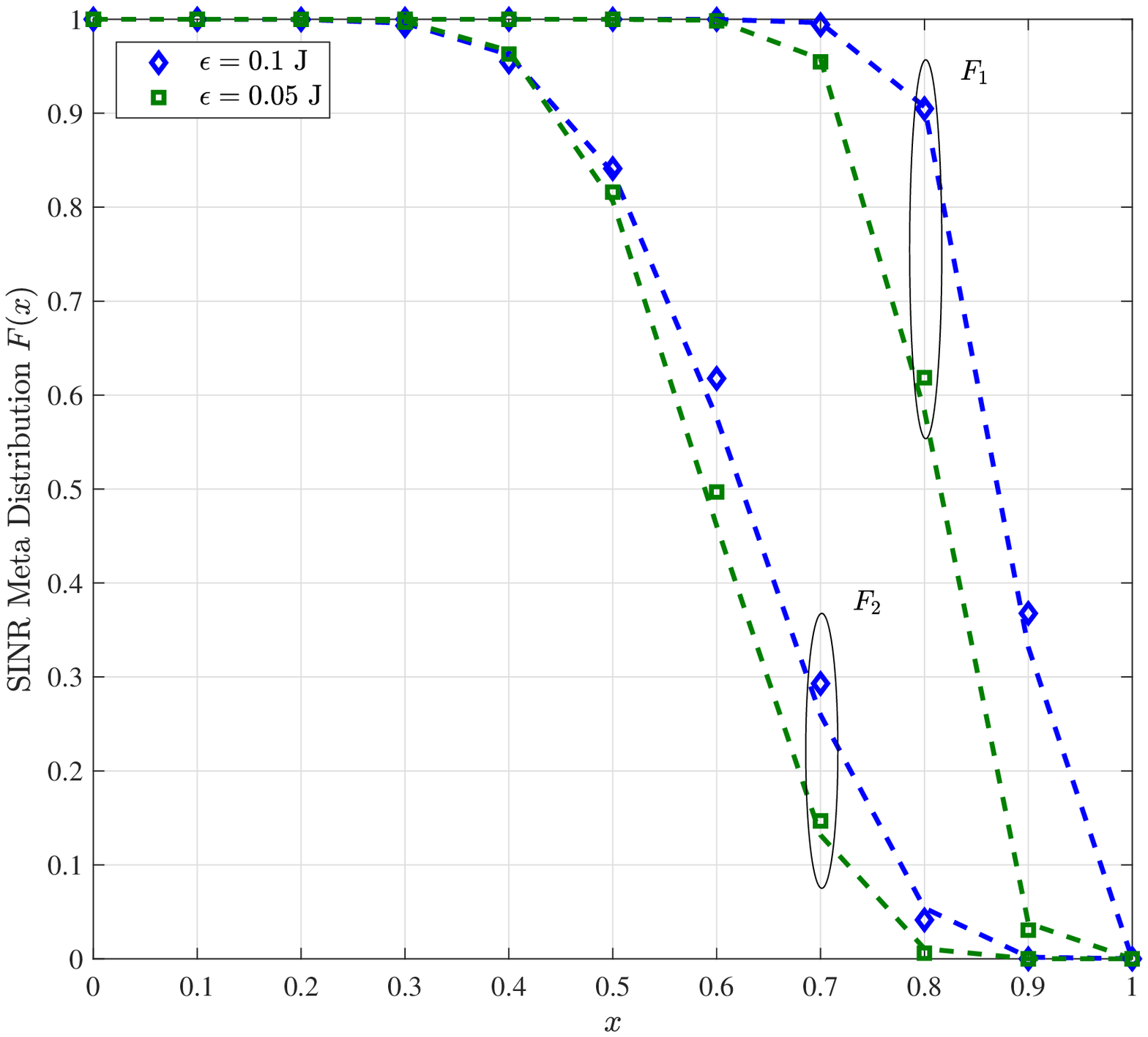}
	\caption{SINR meta Distribution versus threshold x; markers and dashed lines correspond to simulation and analytical results respectively. $\zeta=-10$ dB, $\lambda_1=\SI[per-mode=symbol]{0.1}[ ]{\per\square\metre\per\second}$, $\lambda_2=\SI[per-mode=symbol]{1}[ ]{\per\square\metre\per\second}$, $P_1=1$\si{\watt}, $\rho=2$ \si{\meter} and $\mathcal{E}=0.5$ \si{\joule}.}\label{meta2}
\end{figure}
In Fig. \ref{meta}, we plot the meta distribution of the SINR with respect to the threshold $x$, for the primary and secondary network, with and without cognition. As expected, $F(x)$ decreases as the threshold $x$ increases. In addition, for both the primary and the secondary network, we can make similar observations with Fig. \ref{cover}. When no cognition is considered ($\rho = 0$ \si{\meter}), the interference from the secondary network is stronger resulting in a lower coverage probability and subsequently lower $F(x)$ is obtained. In contrast, when the guard zone radius is $\rho = 2$ \si{\meter}, a higher coverage probability is achieved resulting in a higher $F(x)$. Note that, different from the coverage probability, which is the cumulative CDF of the SINR, the meta distribution of the SINR $F(x)$ consists of a more reliable performance metric of the network, since it express the fraction of receivers that achieve a given threshold $\zeta$ with a given success rate i.e., the coverage probability. Moreover, we present the meta distribution for lower SINR threshold i.e., $\zeta=-10$ dB, in Fig. \ref{meta2}, for energy thresholds $\epsilon=0.05$ and $\epsilon=0.1$. As expected, since the SINR threshold is lower, the meta distribution performs better in comparison with Fig. \ref{meta}. In addition, when $\epsilon=0.05$, the energy coverage probability, and subsequently the transmit probability of the secondary transmitters are both higher. As a result, the density of the active secondary transmitters $\lambda_2^a$ increases, yielding to stronger interference power, which results in a lower meta distribution. As such, with a lower energy threshold, even though the density of the active secondary transmitters increases, the population of the pairs (both primary and secondary) that achieve the required SINR, decreases.

\section{Conclusion}\label{Conclusion}
In this paper, we studied an ad hoc cognitive secondary network which is wirelessly powered by ambient RF signals and is underlaid with an ad hoc primary network. We considered asynchronous channel access from both the networks in order to capture the sporadic channel traffic in IoT environments. For this purpose, the two networks were modeled by TS-PPP and by exploiting tools from stochastic geometry we provided rigorous mathematical analysis for the evaluation of the system's performance. We derived closed form expressions for the transmission probability of the secondary transmitters as well as the spatial throughput and the meta distribution of the SINR for both primary and secondary networks. We presented numerical results which validated our analysis and provided general insights on the main system parameters. We discussed the trade-off between the energy coverage and transmission probability regarding the choice of the information and energy harvesting slots and we showed the subsequent effect on the spatial throughput performance. Finally, our results prove that with cognition, the two networks can coexist with the proper combination of densities and guard zone area to obtain an adequate quality of service. 

\appendix
\subsection{Proof of Lemma \ref{characteristic}}\label{app1}
The characteristic function of $E_H$ is evaluated as follows
\begin{align}
&\mathcal{F}(z) = \mathbb{E}[\exp(\imath z E_H)]\nonumber\\
&= \mathbb{E}\left[\exp\left(\imath z \int_0^{T_E} y(t)\,dt\right)\right]\nonumber\\
&= \mathbb{E}\left[\exp\left(\imath z P_1 \sum_{(x,t) \in \Phi_1} h_x (1+r_x^{\alpha})^{-1} \int_0^{T_E} \phi(\tau,t) \,d\tau\right)\right]\nonumber\\
&= \mathbb{E}_{\Phi_1} \left[\mathbb{E}_{h_x} \left[\exp\left(\imath z P_1 \sum_{(x,t) \in \Phi_1} h_x (1+r_x^{\alpha})^{-1} \psi(t)\right)\right]\right],
\end{align}
where $\psi(t)=\int_0^{T_E}\phi(\tau,t)\,d\tau$ and is given as follows
	\begin{itemize}
		\item if $T_E>T_I$,
		\begin{equation}\label{psi1}
			\psi(t)= \begin{cases}
				T_I+t, & -T_I<t\leq 0,\\
		T_I, & 0<t\leq T_E-T_I,\\
		T_E-t, & T_E-T_I<t<T_E,	
			\end{cases}
		\end{equation}
		\item if $T_E=T_I$,
				\begin{equation}\label{psi2}
				\psi(t)= \begin{cases}
				T_I+t, & -T_I<t<0,\\
				T_I-t, & 0<t<T_I,\\
				T_I, & t=0,				
				\end{cases}
				\end{equation}		
		\item if $T_E<T_I$,
		\begin{equation}\label{psi3}
		\psi(t)= \begin{cases}
		T_I+t, & -T_I<t\leq T_E-T_I,\\
		T_E, & T_E-T_I<t\leq 0,\\
	    T_E-t, &  0<t< T_E-T_I.	
		\end{cases}
		\end{equation}
	\end{itemize}
	
\noindent Then, from the moment generating function of an exponential random variable, since $h_x$ are independent and identically distributed exponential random variables we get
\begin{equation}
\mathcal{F}(z)=\mathbb{E}_{\Phi_1} \left[\prod_{(x,t) \in \Phi_1} \frac{1}{1 -\imath z P_1 \psi(t) (1+r_x^{\alpha})^{-1}}\right].
\end{equation}

\noindent Finally, by using the probability generating functional of a PPP \cite{HAE} and since $\Phi_1$ is a TS-PPP we take the expected value over both, time and locations as follows
\begin{align}
&\mathcal{F}(z) = \exp\Bigg(-2 \pi \lambda\nonumber\\
&\times \int_{-T_I}^{T_E} \int_0^\infty \left(1-\frac{1}{1 - \imath z P_1 \psi(t) (1+u^{\alpha})^{-1}}\right)u \,du\,dt\Bigg).\label{cqapp2}
\end{align}
Note that, the integral limits $[-T_I,T_E]$ take into account all the transmitters which are active during the time interval $[0,T_E]$. That is, we take into account transmissions which have been already initiated before time zero and continue transmitting during the harvesting period. Recall that, each primary transmitter transmits for a period $T_I$ which explains the lower limit, and the harvesting slot lasts for a period $T_E$ which explains the upper limit.

\subsection{Proof of Proposition \ref{sec_power}}\label{E EH}
From equations \eqref{power} and \eqref{transmit power}, the average transmit power of the secondary transmitters is evaluated as
\begin{align}
P_2 &= \frac{\mathbb{E}[E_H]}{T_I} (\mathbb{P}(E_H \leq  \mathcal{E}) - \mathbb{P}(E_H \leq \epsilon)) + \frac{\mathcal{E}}{T_I} \mathbb{P}(E_H > \mathcal{E})\nonumber\\
&=\frac{\mathbb{E}[E_H]}{T_I} (\pi(\mathcal{E}) - \pi(\epsilon)) + \frac{\mathcal{E}}{T_I} \pi(\mathcal{E}),
\end{align}
where $\pi(x)$ is the complementary CDF of the random variable $E_H$ provided in Proposition \ref{E CDF}. We now need to calculate the expected value of $E_H$ as follows.
\begin{align}\label{mean}
\mathbb{E}[E_H] &= \int_0^{T_E} \mathbb{E}[y(t)]\,dt \nonumber\\
&\stackrel{(a)}{=} \int_0^{T_E} \mathbb{E}\left[P_1\sum_{(x,t_x) \in \Phi_1} \phi(t,t_x) (1+r_x^a)^{-1}\right]\,dt\nonumber\\
&\stackrel{(b)}{=} 2 \pi \lambda_1 P_1 \int_0^{T_E} \int_{-\infty}^\infty \int_0^\infty \phi(t,\tau) (1+r^a)^{-1} r \,dr\,d\tau\,dt\nonumber\\
&\stackrel{(c)}{=} 2 \pi^2 \lambda_1 T_E T_I \frac{P_1}{\alpha} \csc\left(\frac{2\pi}{\alpha}\right),
\end{align}
where $(a)$ occurs by using the fact that $h_x$ are independent and identically distributed exponential random variables with unit variance; $(b)$ follows from the Campbell's theorem and $(c)$ follows from the transformation $x^\alpha \to u$ and \cite[3.194.4]{GRAD}.

\subsection{Proof of Theorem \ref{thm1}}\label{app2}
By substituting $\text{SINR}_1$ given by \eqref{sir1}, the coverage probability is evaluated as follows 
\begin{align}
&p_1^c(\zeta) = \mathbb{P}\left(\frac{P_1 h_0 (1+d^a)^{-1}}{\mathrm{I}_1+\mathrm{I}_2+\sigma^2} \geq \zeta\right)\nonumber\\
&= \mathbb{P}\left(h_0 \geq \zeta \frac{(\mathrm{I}_1 + \mathrm{I}_2 + \sigma^2) (1+d^{\alpha})}{P_1}\right)\nonumber\\
&= \mathbb{E}\left[\exp\left(-\zeta \frac{\mathrm{I}_1 (1+d^{\alpha})}{P_1}\right)\right] \mathbb{E}\left[\exp\left(-\zeta\frac{\mathrm{I}_2 (1+d^{\alpha})}{P_1}\right)\right] \nonumber\\
&\times \exp\left(-\zeta\frac{\sigma^2 (1+d^{\alpha})}{P_1}\right),\label{h}
\end{align}
which follows from the complementary CDF of an exponential random variable since $h_0 \sim \exp(1)$. We now set $s = \frac{\zeta (1+d^{\alpha})}{P_1}$ and denote by $\mathcal{L}_{I_n}(s) \triangleq \mathbb{E}[\exp(-\mathrm{I}_n s)]$ the Laplace transform of interference term $\mathrm{I}_n$ evaluated at $s$.

From \eqref{timeaveragedI1} and \eqref{timeaveragedI2}, we obtain the time averaged interference from the primary and secondary transmitters $\mathrm{I}_1$ and $\mathrm{I}_2$, by substituting \eqref{I1} and \eqref{I2}. Then, we have
\begin{align}
&\mathrm{I}_1 = P_1 \sum_{(x,t_x) \in \Phi_1} \chi(t_x) h_x (1+r_x^{a})^{-1},\\
&\mathrm{I}_2 = P_2 \sum_{(y,t_y) \in \Phi_2^a} \chi(t_y) g_y (1+w_y^{a})^{-1},\label{int2}
\end{align}
where \cite{Aloha}
\begin{equation}\label{indicator}
\chi(t_i) \triangleq \frac{1}{T_I} \int_0^{T_I} \phi(t,t_i) \,dt =
\begin{cases}
\frac{T_I - \vert t_i \vert}{T_I}, & t_i \in [-T_I,T_I],\\
0, & \text{otherwise},
\end{cases}
\end{equation}
and $\phi(t,t_i) \triangleq \mathbf{1}(t_i \leq t \leq t_i+T_I)$. Therefore, the Laplace transform of $\mathrm{I}_1$ and $\mathrm{I}_2$ is calculated as
\begin{align}
&\mathcal{L}_{I_n}(s) = \mathbb{E}[\exp(-\mathrm{I}_n s)]\nonumber\\
&=\mathbb{E}_{\Phi_n} \left[\mathbb{E}_{h_x} \left[\exp\left(-s P_n \sum_{(x,t_x) \in \Phi_n} \chi(t_x) h_x (1+r_x^{\alpha})^{-1}\right) \right]\right]\nonumber\\
&\stackrel{(a)}{=} \mathbb{E}_{\Phi_n}\left[\prod_{(x,t_x) \in \Phi_n} \frac{1}{1 + s P_n \chi(t_x) (1+r_x^{\alpha})^{-1}} \right]\nonumber\\
&\stackrel{(b)}{=} \exp\Bigg(-2\pi\lambda_n \nonumber\\
&\times\int_{-T_I}^{T_I} \int_0^\infty \left(1-\frac{1}{1 + s P_n \chi(t) (1+u^{\alpha})^{-1}}\right)u du\,dt\Bigg),\nonumber\\
&\stackrel{(c)}{=}\exp\Bigg(-2 \lambda_n \frac{\pi^2 P_n s}{\alpha}\csc \left(\frac{2 \pi}{\alpha}\right)\nonumber\\ &\quad \quad \quad \quad \quad \quad \times\int_{-T_I}^{T_I}\chi(t)\left(1+\chi(t)P_n s\right)^{\frac{2-\alpha}{\alpha}}\, dt \Bigg)\label{eqapp2},
\end{align}
where $(a)$ follows from the moment generating function of an exponential random variable and since $h_x$ are independent and identically distributed exponential random variables; $(b)$ make use of the probability generating functional of a PPP \cite{HAE}, thus unconditioning on both time and space; $(c)$ is obtained by \cite[3.222.2]{GRAD} and the evaluation of the last integral is straightforward. Then, the final result follows by substituting \eqref{eqapp2} in \eqref{h}.

\end{document}